\documentclass[submission,copyright,creativecommons]{eptcs}
\providecommand{\event}{ACT 2023} 

\usepackage{iftex}

\ifpdf
  \usepackage{underscore}         
  \usepackage[T1]{fontenc}        
\else
  \usepackage{breakurl}           
\fi

\title{Syntax Monads for the Working Formal Metatheorist}

\author{
  Lawrence Dunn\institute{University of Pennsylvania\\Philadelphia, USA}
  \email{dunnla@seas.upenn.edu}
  \and
  Val Tannen\institute{University of Pennsylvania\\Philadelphia, USA}
  \email{val@cis.upenn.edu}
  \and
  Steve Zdancewic\institute{University of Pennsylvania\\Philadelphia, USA}
  \email{stevez@cis.upenn.edu}
}

\newcommand{\titlerunning}{Syntax Monads for the Working Metatheorist}
\newcommand{\authorrunning}{L. Dunn, V. Tannen \& S. Zdancewic}

\hypersetup{
  bookmarksnumbered,
  pdftitle    = {\titlerunning},
  pdfauthor   = {\authorrunning},
  pdfsubject  = {Decorated traversable monads for syntax metatheory} ,               
  pdfkeywords = {first-order representations, string diagrams, distributive laws, bialgebras, metatheory} 
}

\usepackage{amssymb,amsmath,amsthm, amsfonts}
\usepackage{mathtools} 
\usepackage{multirow}  
\usepackage{floatrow}  
\usepackage{dsfont}    
\usepackage{array}
\usepackage{subcaption}
\usepackage{url}

\usepackage{tikz}
\usetikzlibrary{cd}
\usepackage{includes/tikzit}

\definecolor{teared}{RGB}  {214,  39, 40}
\definecolor{teablue}{RGB} { 31, 119, 180}
\definecolor{teagreen}{RGB}{ 44, 160,  44}
\definecolor{hired}{RGB}  {220,  38, 127}
\definecolor{hiblue}{RGB} {100, 143, 255}
\definecolor{hiyellow}{RGB}{ 255, 176,  0}

\newcommand{\colorMonad}{hiblue}
\newcommand{\colorComonad}{hired}
\newcommand{\colorApplicative}{hiyellow}

\tikzstyle{MonadNode}=   [draw=\colorMonad,   fill=\colorMonad,   shape=circle]
\tikzstyle{ComonadNode}= [draw=\colorComonad, fill=\colorComonad, shape=circle]
\tikzstyle{DecMonNode}=  [draw=\colorMonad,   fill=\colorComonad, line width=1.4pt, shape=circle]
\tikzstyle{DecPlainNode}=[draw=black,         fill=\colorComonad, line width=1.4pt, shape=circle]

\tikzstyle{Identity}=       [-, draw=black, dashed,     line width=1.4pt]
\tikzstyle{Plain}=          [-, draw=black,             line width=1.4pt]
\tikzstyle{MonadWire}=      [-, draw=\colorMonad,       line width=1.4pt]
\tikzstyle{DecorationWire}= [-, draw=\colorComonad,     line width=1.4pt]
\tikzstyle{ApplicativeWire}=[-, draw=\colorApplicative, line width=1.4pt]

\tikzstyle{Box}=            [fill=white, draw=black, shape=rectangle, line width=1pt]
\tikzstyle{Two Input Box}=  [fill=white, draw=black, shape=rectangle, minimum height=20pt, line width=1pt]
\tikzstyle{Three Input Box}=[fill=white, draw=black, shape=rectangle, minimum height=80pt, line width=1pt]
\tikzstyle{Equals}=         [fill=none,  draw=none, font={\Huge}, shape=circle]
\tikzstyle{Def}=            [fill=none,  draw=none, font={\Large}, shape=circle]
\tikzstyle{Label}=          [fill=white, draw=none, shape=circle, font={\large}]

\tikzstyle{Frame}=[-, line width=1.4pt, draw=none, dashed]
\tikzstyle{Highlight}=[-, draw=black, dashed]

\newtheorem{theorem}{Theorem}[section]
\newtheorem{lemma}[theorem]{Lemma}

\newtheorem{definition}[theorem]{Definition}

\newtheorem{example}[theorem]{Example}
\numberwithin{equation}{section}

\newcommand{\sdiagram}[2]{\scalebox{#1}{\input{#2.tikz}}}

\newcommand{\id}[1]{\mathrm{id}_{\msub{#1}}}

\newcommand{\coq}[1]{{\normalfont \texttt{#1}}} 
\newcommand{\coqlist}{\coq{list}}
\newcommand{\smallcoq}[1]{{\texttt{\scriptsize #1}}} 
\newcommand{\idfun}{\mathds{1}}
\newcommand{\comp}{\cdot}
\newcommand{\fn}[1]{{#1}}
\newcommand{\SigmaLambda}[1]{\Sigma^\lambda_{#1}}

\newif\ifshowsubs
\showsubstrue
\newcommand{\maybesubscript}[1]{\ifshowsubs{#1} \else {} \fi}
\newcommand{\msub}[1]{\maybesubscript{#1}}

\newcommand{\extr}[2]{\mathrm{extr}^{#1}_{\msub{#2}}}
\newcommand{\del}[2]{\mathrm{del}^{#1}_{\msub{#2}}}
\newcommand{\dup}[2]{\mathrm{dup}^{#1}_{\msub{#2}}}

\newcommand{\strength}[3]{\mathrm{st}^{#1}_{#2\msub{, #3}}}

\newcommand{\map}[2]{\mathrm{map}^{#1}_{\msub{#2}}}
\newcommand{\ret}[2]{\mathrm{ret}^{#1}_{\msub{#2}}}
\newcommand{\join}[2]{\mathrm{join}^{#1}_{\msub{#2}}}
\newcommand{\dec}[2]{\mathrm{dec}^{#1}_{\msub{#2}}}
\newcommand{\dist}[3]{\mathrm{dist}^{#1}_{#2\msub{, #3}}}
\newcommand{\pure}[2]{\mathrm{pure}^{#1}_{\msub{#2}}}

\newcommand{\bind}[2]{\mathrm{bind}^{#1}_{\msub{#2}}}
\newcommand{\binddt}[3]{\mathrm{binddt}^{#1}_{#2\msub{, #3}}}

\newcommand{\Set}{\mathbf{Set}}
\newcommand{\EndSet}{\mathbf{End}_{\mathrm{Set}}}
\newcommand{\Dec}[1]{\mathbf{Dec}_{#1}}
\newcommand{\Trav}{\mathbf{Trav}}
\newcommand{\DecTrav}[1]{\mathbf{DecTrav}_{#1}}

\newcommand{\funF}{\mathrm{F}}
\newcommand{\funG}{\mathrm{G}}
\newcommand{\funT}{\mathrm{T}}
\newcommand{\typA}{\mathrm{A}}
\newcommand{\typB}{\mathrm{B}}

\newcommand{\typE}{\mathrm{E}}
\newcommand{\typW}{\mathrm{W}}

\newcommand{\writer}{\typW \times}
\newcommand{\writerLong}{\typW \times -}
\newcommand{\writerShort}{\typW^\times}

\newcommand{\envShort}{\typE^\times}

\newcommand{\term}[1]{\coq{term}_{#1}}
\newcommand{\App}{\coq{App}}
\newcommand{\Abs}{\coq{Lam}}
\newcommand{\Var}{\coq{Var}}

\newcommand{\fvar}{\coq{fvar}}
\newcommand{\bvar}{\coq{bvar}}
\newcommand{\LNvar}{\mathbb{A} + \mathbb{N}}

\newcommand{\LC}{\mathrm{LC}}
\newcommand{\LCloc}{\mathrm{LC}_{\mathrm{loc}}}
\newcommand{\open}{\mathrm{open}}
\newcommand{\openloc}{\mathrm{open}_{\mathrm{loc}}}
\newcommand{\subst}{\mathrm{subst}}
\newcommand{\substloc}{\mathrm{subst}_{\mathrm{loc}}}

\newcommand{\axiomMonadUnitLeft}{%
\join{\funT}{\typA} \comp \ret{\funT}{\funT \typA} = %
\id{\funT \typA}%
}
\newcommand{\axiomMonadUnitRight}{%
\join{\funT}{\typA} \comp \map{\funT}{} \left(\ret{\funT}{\typA}\right) = %
\id{\funT\typA}%
}
\newcommand{\axiomMonadUnitAll}{%
\join{\funT}{\typA} \comp \ret{\funT}{\funT \typA} = %
\id{\funT\typA} =%
\join{\funT}{\typA} \comp \map{\funT}{} \left(\ret{\funT}{\typA}\right)%
}
\newcommand{\axiomMonadAssoc}{%
\join{\funT}{\typA} \comp \join{\funT}{\funT\typA} = %
\join{\funT}{\typA} \comp \map{\funT}{} \left(\join{\funT}{\typA}\right)%
}

\newcommand{\axiomdecorationextract}{%
  \map{\funT}{\writerShort{\typA}, \typA} \extr{\writer}{\typA} \comp \dec{\funT}{\typA} =%
  \id{\funT\typA}%
}
\newcommand{\axiomdecorationduplicate}{%
  \dec{\funT}{\writer \typA} \comp \dec{\funT}{\typA} =%
  \map{\funT}{\writer{\typA}, \writer\writer{\writer{\typA}}} \dup{\writer}{\typA} \comp \dec{\funT}{\typA}%
}

\newcommand{\axiomDecorationRet}{%
  \dec{\funT}{\typA} \comp \ret{\funT}{\typA} = %
  \ret{\funT}{\writer\typA} \comp \ret{\writer}{\typA}%
}
\newcommand{\axiomDecorationJoin}{%
  \dec{\funT}{\typA} \comp \join{\funT}{\typA} = %
  \join{\funT \comp \writerShort}{}%
  \comp \dec{\funT}{\funT(\writer\typA)}%
  \comp \map{\funT}{\funT \typA, \funT (\writer{\typA}) } \left(\dec{\funT}{\typA}\right)%
}

\newcommand{\axiomDistUnit}{%
\dist{\funT}{\idfun}{\typA} = \id{\funT\typA}%
}
\newcommand{\axiomDistComp}{%
\dist{\funT}{\funF \comp \funG}{\typA} = %
\map{\funF}{} \left(\dist{\funT}{\funG}{\typA}\right) \comp \dist{\funT}{\funF}{\funG\typA}%
}
\newcommand{\axiomDistMorphism}{%
\dist{\funT}{\funG}{\typA} \comp \map{\funT}{} \left(\phi_{\typA}\right) =%
\phi_\typA \comp \dist{\funT}{\funF}{\typA}
}

\newcommand{\axiomDistRet}{%
\dist{\funT}{\funF}{\typA} \comp \ret{\funT}{\funF \typA} = %
\map{\funF}{} \left(\ret{\funT}{\typA}\right)%
}
\newcommand{\axiomDistJoin}{%
\dist{\funT}{\funF}{\typA} \comp \join{\funT}{\funF \typA} =%
\map{\funF}{} \left(\join{\funT}{\typA}\right) \comp \dist{\funT}{\funF}{\funT\typA} \comp \map{\funT}{} \left(\dist{\funT}{\funF}{\typA}\right)
}

\newcommand{\traversableMorphism}{%
\dist{\funT_2}{\funF}{\typA} \comp \psi_{\msub{\funF \typA}} =%
\map{\funF}{} \left(\psi_{}\right) \comp \dist{\funT_1}{\funF}{\typA}%
}
\newcommand{\axiomDecorationTraverse}{%
\map{\funF}{} \left(\dec{\funT}{\typA}\right) \comp \dist{\funT}{\funF}{\typA} =%
\dist{\funT}{\funF}{\writerShort\typA} \comp \map{\funT}{} \left(\dist{\writerShort}{\funF}{\typA}\right) \comp \dec{\funT}{\funF\typA}%
}

\newcommand{\defDecoratedMorphism}{%
\dec{\funT_2}{} \comp \phi =%
\phi \comp \dec{\funT_1}{}%
}
\newcommand{\defDecoratedId}{%
\dec{\idfun}{} =%
\ret{\writerShort}{}%
}
\newcommand{\defDecoratedComp}{%
\dec{\funT_1 \comp \funT_2}{} =%
\map{\funT_1}{} \left(\map{\funT_2}{}\, \left(\join{\writerShort}{}\right) \comp \strength{\funT_2}{\typW}{}\right) \comp \dec{\funT_1}{} \comp \map{\funT_1}{} {\dec{\funT_2}{}}{}%
}

\newcommand{\stringdiagramequation}[3]{%
	\noindent\begin{minipage}{.39\linewidth}%
		\centering%
		\sdiagram{0.6}{tikz/#1}%
	\end{minipage}%
	\begin{minipage}{.59\linewidth}%
		\begin{equation}%
			#2 \label{eqn:#3-string}%
		\end{equation}%
	\end{minipage}%
}
\newcommand{\stringdiagramequationParam}[6]{%
	\noindent\begin{minipage}{#1\linewidth}%
		\centering%
		\sdiagram{#3}{tikz/#4}%
	\end{minipage}%
	\begin{minipage}{#2\linewidth}%
		\begin{equation}%
			#5 \label{eqn:#6-string}%
		\end{equation}%
	\end{minipage}%
}

\newcommand{\stringdiagramequationBig}[3]{%
	\noindent\begin{minipage}{1\linewidth}%
		\centering%
		\sdiagram{0.6}{tikz/#1}%
	\end{minipage}\\
	\begin{minipage}{1\linewidth}%
		\centering
		\begin{equation}%
			#2 \label{eqn:#3-string}%
		\end{equation}%
	\end{minipage}%
}

\newcommand{\stringDiagramOperationSide}[2]{%
	\noindent\begin{minipage}{.5\linewidth}%
		\centering%
		\sdiagram{0.8}{tikz/#1}%
	\end{minipage}%
	\begin{minipage}{.5\linewidth}%
		\begin{equation*}%
			#2%
		\end{equation*}%
	\end{minipage}%
}

\begin{document}

\maketitle

\begin{abstract}
  Formally verifying the properties of formal systems using a proof
  assistant requires justifying numerous minor lemmas about
  capture-avoiding substitution. Despite work on category-theoretic
  accounts of syntax and variable binding, \emph{raw, first-order}
  representations of syntax, the kind considered by many practitioners
  and compiler frontends, have received relatively little
  attention. Therefore applications miss out on the benefits of
  category theory, most notably the promise of reusing formalized
  infrastructural lemmas between implementations of different
  systems. Our Coq framework Tealeaves provides libraries of reusable
  infrastructure for a raw, locally nameless representation and can be
  extended to other representations in a modular fashion. In this
  paper we give a string-diagrammatic account of \emph{decorated
  traversable monads} (DTMs), the key abstraction implemented by
  Tealeaves. We define DTMs as monoids of structured endofunctors
  before proving a representation theorem à la Kleisli, yielding a
  recursion combinator for finitary tree-like datatypes.
\end{abstract}

\section{Introduction}
\label{introduction}

Machine-certified proofs of the properties of programming languages,
type theories, and other formal systems are increasingly critical for
establishing confidence in the design and implementation of computer
systems. Much of this reasoning is overtly concerned with the
manipulation of syntactical structures, especially variable-binding
constructs, making the representation of these structures a key issue
in formal metatheory \cite{2005:poplmark}. As implementations scale in
complexity to realistic formalizations of compilers \cite{2012:vellvm}
and programming languages \cite{2014:cakeml}, often with many kinds of
variables, the bookkeeping required to manipulate variables correctly
becomes nearly prohibitive.

Category-theoretic accounts of syntax with variable binding
(e.g. \cite{1994:bh-formal-monads, 1999:abstractsyntax,
  2008:fiore-soas, 2022:fiore-metatheory-soas,
  2022:ahrens-category-typed-syntax}) offer the tantalizing benefit of
formalizing tedious syntax ``infrastructure'' once and for all over an
abstract choice of signature, instead of repeating this effort for the
particular syntax of each new system. However, the kind of syntax
usually considered by theorists---often intrinsically well-typed with
well-scoped de Bruijn indices---is different from what many working
semanticists and compilers actually implement. Consequently, the
benefits of a principled categorical framework are not yet available
to many applications. This work lays the foundations of a
category-theoretic account of variable binding as it often looks in
practice, with the aim of building certified libraries of generic
syntax infrastructure that can be used (and reused) in real-world
applications.

\noindent \textbf{Contributions}. \enspace This manuscript makes two contributions.
\begin{itemize}
\item We introduce the strict monoidal category $\textbf{DecTrav}_{W}$
  of decorated-traversable endofunctors on $\Set$ for some monoid $W$
  (Definition \ref{def:dectrav}) and define decorated-traversable
  monads (DTMs) as monoids in this category (Definition
  \ref{def:dtm}). Examples of decorated-traversable functors come from
  the signature functors of languages with variable binding; the free
  monads they generate are DTMs. These structures admit a
  string-diagram calculus, which aids in equational reasoning.
\item We prove an equivalence (Theorem \ref{thm:dtmtokleisli}) between
  monoids in $\textbf{DecTrav}_{W}$ and a Kleisli-style presentation
  (Definition \ref{def:dtmkleisli}) that describes a
  structured recursion combinator for abstract syntax trees.
\end{itemize}

As with ordinary (strong) monads \cite{1988:moggi-monads}, the Kleisli
presentation is of more immediate utility from a functional
programming or formal metatheory perspective, in part because the
definition requires checking fewer axioms. In a previous,
tool-oriented paper \cite{2023:tealeaves} we introduced
Kleisli-presented DTMs and used them to derive generic syntax
infrastructure for first-order representations of variable binding in
Coq. However, that paper did not explain why the seemingly ad-hoc
equational axioms should be considered ``correct.'' This paper
justifies the robustness of the axioms by proving their equivalence
with a more clearly principled, string-diagrammatic set of axioms. The
results in this paper have been formalized in Coq and are available in
our GitHub
repository.\footnote{\url{https://github.com/dunnl/tealeaves}}

\noindent \textbf{Layout}. \enspace The rest is laid out as
follows. Section \ref{sec:fo-syntax} contains background on
first-order representations of variable binding. We recall that
abstract syntax trees, parameterized by the data in the leaves,
naturally give rise to a (free) monad.  For such monads, the Kleisli
axiomatization provides a theory of na\"ive substitution, but this is
not expressive enough to define the \emph{capture-avoiding}
substitution operations considered by different representations of
variables. Section \ref{sec:decorated-traversable-functors} introduces
the endofunctor categories $\Dec{W}$, $\Trav$, and
$\DecTrav{W}$. Section \ref{sec:kleisli} derives a Kleisli-style
characterization of monoids in $\DecTrav{W}$ and explains why this
abstraction solves the problems identified in Section
\ref{sec:fo-syntax}. Section \ref{sec:related} contrasts our approach
with related work. Section \ref{sec:conclusion} concludes.

Functors in this paper have type $\Set \to \Set$ and typically
represent parameterized container types like lists, binary trees, and
abstract syntax trees. We recall that $\EndSet$ is the strict monoidal
category whose objects are endofunctors on $\Set$, whose arrows are
natural transformations, and whose tensor product is given by
composition of functors.

\section{First-order Representations of Variable Binding}
\label{sec:fo-syntax}

The modern formal metatheorist has many options for representing and
manipulating terms with variable binding in a proof assistant. The
first choice is whether to employ a \emph{first-order} or
\emph{higher-order} approach. Higher-order strategies represent
variable-binding constructors in the object language as higher-order
functions in the metatheory; this sidesteps thorny issues like
variable capture but does not shed much light on syntax as defined in,
say, a compiler. We are interested in things like verified compilers,
so we consider a first-order approach. This style is also simple,
intuitive, and well supported by general-purpose proof assistants like
Coq \cite{2023:coq}. Theoretically it lends itself to the
theory of initial algebras, the category-theoretic take on structural
recursion \cite{1969:burstall-struct-rec}.

A more or less orthogonal question is whether to consider an
\emph{intrinsic} or \emph{extrinsic} (also called \emph{raw})
representation. For instance, intrinsically well-scoped terms exist in
some context $\Gamma$ and can only mention free variables declared in
$\Gamma$, while instrinsically well-typed terms essentially carry
around their own typing judgment. The raw approach posits that a
single set of terms simply exists, including ones that are ill-formed
and untypable in the formal system. Properties like being well-scoped
in $\Gamma$ are then defined post-hoc as predicates on terms by
structural recursion. We consider an extrinsic representation, though
in future work we could investigate an intrinsic approach.

Finally, one has a choice about how to represent free and bound
variables, i.e. the datatype stored in the \emph{leaves} of syntax
trees. Encoding strategies go by names like \emph{fully named},
\emph{de Bruijn indices}, \emph{de Bruijn levels}, \emph{locally
named}, \emph{locally nameless}, and variations. DTMs capture what is
tree-like about syntax without saying anything about the type of data
in the leaves, and for now we shall remain agnostic about this
choice. Figure \ref{fig:lam-syntax} displays a first-order definition
of the set of raw lambda terms. The only unusual part of this
definition is that we parameterize the set of terms by a
representation of variables $V$ and binder annotations $B$. These
parameters will be fixed by a variable encoding strategy in Section
\ref{subsec:variable-encodings}.

\noindent
\begin{minipage}[b]{0.53\linewidth}
\vspace{-0.5\baselineskip}
  \begin{figure}[H]
\begin{verbatim}
Inductive term (B V : Set) : Set :=
  | Var: V -> term B V
  | App: term B V -> term B V -> term B V
  | Lam: B -> term B V -> term B V.
\end{verbatim}
\caption{Syntax of the lambda calculus in Coq}
\label{fig:lam-syntax}
  \end{figure}
\end{minipage}
\begin{minipage}[b]{0.45\linewidth}
\vspace{-0.5\baselineskip}
  \begin{figure}[H]
    \[\arraycolsep=1pt
    \begin{array}{lcl}
      \bind{}{}\,f\,\left(\Var\,v \right)     &=& f v \\
      \bind{}{}\,f\,\left(\App\,t_1\,t_2\right) &=& \App\,\left(\bind{}{}\,f\,t_1\right) \left(\bind{}{}\,f\,t_2\right) \\
      \bind{}{}\,f\,\left(\Abs\,b\,t\right)     &=& \Abs\,b\,\left(\bind{}{}\,f\,t\right) \\
    \end{array}
    \]
    \caption{$\mathrm{bind}$ instance for \coq{term}}
    \label{fig:lam-bind}
  \end{figure}
\end{minipage}

To concentrate on $\term{}$ as functor in $V$, we shall typeset $B$ as
a subscript. Associated to the lambda calculus is a signature functor
\[\SigmaLambda{B} X \overset{def}{=} X \times X + B \times X\]
encoding the domain of the two constructors of $\term{}$ besides
\coq{Var}.  $\term{B} V$ is defined as the least fixpoint $\mu
X. \left(V + \SigmaLambda{B} X\right)$, i.e. as the smallest
solution to the following equation:
\[
\coq{term}_B\, V \simeq V + \coq{term}_B\, V \times \coq{term}_B\, V + B \times \coq{term}_B\, V.
\]
It is well known that, by its least fixed point construction, a
datatype like $\term{B}$ (for any $B$) naturally forms a monad. We
present monads string-diagrammatically alongside a conventional
equational presentation. A general introduction to string diagrams is
outside the scope of this paper, but the interested reader may consult
\cite{fong2018seven, hinzemarsden2023}. In this paper our calculus
depicts a monad $T$ with a blue wire.

\begin{definition}
  \label{def:monad}
  \showsubsfalse
  A \emph{monad} $T$ is a functor equipped with two natural transformations

  \begin{minipage}[b]{.1\linewidth}%
  	\centering
    \sdiagram{0.6}{tikz/monad/ops/unit}
  \end{minipage}%
  \begin{minipage}[b]{.25\linewidth}%
    \begin{equation*}%
      \ret{\funT}{} : \forall \left(A: \Set\right), A \to T A
    \end{equation*}%
  \end{minipage}%
  \begin{minipage}[b]{.2\linewidth}%
  	\centering
    \sdiagram{0.6}{tikz/monad/ops/multiplication}
  \end{minipage}%
  \begin{minipage}[b]{.35\linewidth}%
    \begin{equation*}%
      \join{\funT}{} : \forall \left(A : \Set\right), T (T A) \to T A
    \end{equation*}%
  \end{minipage}%

  subject to the following laws.

  \stringdiagramequation{monad/axioms/identity_l}{\axiomMonadUnitLeft}{monad-id-l}

  \stringdiagramequation{monad/axioms/identity_r}{\axiomMonadUnitRight}{monad-id-r}

  \stringdiagramequation{monad/axioms/associativity}{\axiomMonadAssoc}{monad-assoc}
\end{definition}

The $\ret{\funT}{}$ operation constructs a tree from a single
leaf---for $\term{}$ this is the $\Var$ constructor. $\join{\funT}{}$
flattens a tree-of-trees into a tree by grafting the layers
together. $\map{\funT}{}$ applies a function to each of the
leaves. This presentation is visually pleasing, but fairly abstruse
for our purposes. For applications, the following definition is more
pragmatic.

\begin{definition}
  \label{def:kleislimonad}
  A \emph{Kleisli-presented monad} $T : \Set \to \Set$ is a
  set-forming operation equipped with two polymorphic operations
  \[
  \begin{array}{rcl}
    \ret{}{}  &:& \forall \left(A : \Set\right),\ A \to T A \\
    \bind{}{} &:& \forall \left(A\ B : \Set\right),\ \left(A \to T B\right) \to T A \to T B
  \end{array}
  \]
  subject to the following three laws (implicitly universally quantified over all relevant variables).

  \noindent\begin{minipage}[t]{0.38\linewidth}
  \vspace{-\baselineskip}
    \begin{gather}
      \bind{}{}\ \ret{}{}\ = \id{}   \label{eqn:bind-1} \\
      \bind{}{}\ f \comp \ret{}{} = f \label{eqn:bind-2}
    \end{gather}
  \end{minipage}
  \begin{minipage}[t]{0.6\linewidth}
  \vspace{-\baselineskip}
    \begin{gather}
      \bind{}{}\ g \comp \bind{}{}\ f = \bind{}{}\ \left(\bind{}{}\ g \comp f\right) \label{eqn:bind-3} \\ \notag
    \end{gather}
  \end{minipage}
\end{definition}

The equivalence of these definitions is well-known \cite{Manes1976}.
\begin{lemma}[Manes, 1976]
  \label{lem:monad-to-kleisli}
  Definitions \ref{def:monad} and \ref{def:kleislimonad} are
  equivalent.
\end{lemma}

Figure \ref{fig:lam-bind} gives the $\bind{}{}$ instance for
$\term{}$. We note that $\bind{}{}\ f\ t$ merely applies $f$ to each
variable occurrence in $t$, replacing it with a subterm. We call this
simple replacement operation a \emph{na\"ive}
substitution. \eqref{eqn:bind-1} stipulates that replacing all
variables with themselves yields the original $t$. \eqref{eqn:bind-2}
is the definition of $\bind{}{}$ on $\Var$. \eqref{eqn:bind-3} governs
the composition of multiple substitutions. The limitations of this
na\"ive notion of substitution become apparent when we turn our
attention to situations involving both free and bound variables.

\subsection{Variable Encodings}
\label{subsec:variable-encodings}
We discuss two exemplary techniques for representing variables.

\textbf{Fully named} \quad A fully named approach assigns names,
represented as atoms $a \in \mathbb{A}$, to both free and bound
variables, hence $V = \mathbb{A}$. Variable-binding constructs are
labeled with the names they introduce, so $B = \mathbb{A}$. The set
$\term{\mathbb{A}} \mathbb{A}$ corresponds to the following
pen-and-paper syntax of lambda terms:
\[t ::= a | t t | \lambda a . t\]
Consider the main axiom of lambda calculus, the beta conversion rule
$\left(\lambda x . t_1\right) t_2 =_\beta t_1 \{t_2 / x\}$, where $t_1
\{t_2/x\}$ stands for the capture-avoiding substitution of $t_2$ in
place of free occurrences of $x$ in $t_1$. For instance:
\[
\left(\lambda x . x z\right) \{z/x\} =_\beta \lambda x . x z \quad
\left(\lambda y . xz\right) \{z/x\} =_\beta \lambda y . z z \quad
\left(\lambda z . x z\right) \{z/x\} =_\beta \lambda y . z y \quad
\]
In the first case, $x$ occurs bound and is not replaced, while in the
second and third cases it occurs free and is replaced with $z$. In the
last case, $z$ also happens to be the name of the distinct entity
introduced by the $\lambda$, so a na\"ive substitution would
incorrectly result in the term $\lambda z. z z$. Therefore we rename
this entity, and all variables bound to it, to a non-conflicting name,
say $y$. Renaming variables like this complicates a fully named
representation, and it also complicates the theory of DTMs. Therefore
this manuscript focuses on representations that do not require binder
renaming, but see future work in Section \ref{sec:conclusion}.

\textbf{Locally nameless}\quad The locally nameless strategy
represents free variables as atoms, as before, but represents bound
variables as de Bruijn indices \cite{1972:debruijn-dummies}, natural
numbers that describe the ``distance'' from the occurrence to the
abstraction that introduced it, indexing from $0$. For example,
$\lambda x. \lambda y. x y z$ becomes $\lambda \lambda 1 0 z$. Thus
$V$ is the (tagged) union $\mathbb{A} + \mathbb{N}$. For clarity, we
use \fvar{} and \bvar{} as the names of the left and right injections
(respectively) into $V$.

Because the representation of a bound variable is canonical, there is
no need to give arbitrary names to bound variables, hence no need to
rename them to avoid conflicts. Lambda abstractions do not need to be
annotated with names either, which we formally represent by annotating
them with type $B = \mathbf{1} = \{\star\}$, the singleton. This gives
the set $\term{\mathbf{1}} \left(\mathbb{A} + \mathbb{N}\right)$,
corresponding to the following grammar:
\[t ::= a | n | t t | \lambda t\]
A benefit of locally nameless is that substitution of free variables
is particularly simple: a variable is free exactly when it is an atom,
so it can never be mistaken for a bound variable. Due to this special
simplicity, the ``correct" notion of substitution for free variables,
$\subst$ (Figure \ref{fig:subst}), happens to be expressible using
$\bind{}{}$. This operation has the following type, where
$\subst\ x\ u\ t$ replaces $x$ in $t$ with $u$:
\[\subst : \mathbb{A} \to \term{\mathbf{1}} \left(\LNvar\right) \to \term{\mathbf{1}} (\LNvar) \to \term{\mathbf{1}} \left(\LNvar\right)\]
Figure \ref{fig:bind-subst} defines $\subst$ in terms of $\bind{}{}$
and a function $\substloc$ that prescribes the ``local'' effect of
substitution on individual occurrences. Decomposing $\subst$ like this
practical value because $\substloc$ does not depend on the particulars
of $\term{}$, so this definition is given abstractly over a monad
$T$. This also means we can employ the monad laws to reason about it,
exemplified in the following lemma.

\begin{figure}
  \begin{subfigure}[b]{0.52\linewidth}
    \centering
    \begin{gather*}
      \subst\, x\, u\, \left(\Var\, v\right) =
      \begin{cases*}
        u & \textrm{if $v = \fvar\, x$} \\
        \Var\, v & \textrm{else}
      \end{cases*} \\
      \subst\, x\, u\, \left(\App\, t_1\, t_2\right) = \App\, \left(\subst\, x\, u\, t_1 \right)\, \left(\subst\, x\, u\, t_2 \right) \\
      \subst\, x\, u\, \left(\Abs\, \star\, t\right) =  \Abs\, \star\, \left(\subst\, x\, u\, t \right)
    \end{gather*}
    \caption{Structurally recursive definition}
    \label{fig:subst}
  \end{subfigure}
  \begin{subfigure}[b]{0.47\linewidth}
    \begin{gather*}
      \subst\ x\ u\ t = \bind{}{}\ \left(\substloc\ x\ u\right)\ t \label{eqn:subst-bind} \\
      \substloc\ x\ u\ v = \begin{cases*}
        u & \textrm{if $v = \coq{fvar}\ x$} \\
        \ret{}{}\ v & \textrm{else}
      \end{cases*}
      \vspace{\baselineskip}
    \end{gather*}
    \caption{Definition abstract over a choice of monad}
    \label{fig:bind-subst}
  \end{subfigure}
  \caption{Substitution of atoms in a locally nameless representation}
  \label{fig:subst-ln}
\end{figure}

\begin{lemma}
  \label{lem:basic-subst}
  Let $T$ be any monad, let $x \neq y$ be atoms, and let
  $t [x \mapsto u]$ denote $\mathrm{subst}\ x\ u\ t$, defined abstractly in $T$. Substitution has the following
  properties:\footnote{Where $x$ is used as a term, it is understood
  as the atomic term $\ret{}{} \left(\smallcoq{fvar}\ x\right)$. In the
  third equation, the right side mentions the \emph{parallel}
  substitution that simultaneously replaces all $x$ with $u_1 [y \mapsto u_2]$ and $y$ with $u_2$.}
  \[x [x \mapsto t] = t
  \quad
  t [x \mapsto x] = t
  \quad
  t [x \mapsto u_1][y \mapsto u_2] = t \big[x \mapsto u_1 [y \mapsto u_2]; y \mapsto u_2\big]
  \]
\end{lemma}

Lemma \ref{lem:basic-subst} is easily proven abstractly over $T$ by
appealing to equations \eqref{eqn:bind-1}--\eqref{eqn:bind-3}. On the
other hand, here is a lemma that cannot even be stated, much less
proven, abstractly over $T$:
\begin{lemma}[fresh-subst]
  \label{lem:fresh-subst}
  If an atom $x$ does not occur in $t$, then $t [x \mapsto u] = t$.
\end{lemma}
Lemma \ref{lem:fresh-subst} cannot be formulated abstractly because we
lack a mechanism for defining what it means for an atom to
\emph{occur} in a term---occurrence is a predicate, and $\bind{}{}$
does not provide a mechanism for defining predicates. We can of course
prove the lemma for $\term{}{}$ in particular by structural recursion,
but this is no longer generic over a choice of $T$ and cannot be
shared by users formalizing a different syntax. In order to reason
about syntax as a container (of occurrences of variables) like this,
we define traversable monads in Section \ref{subsec:traversals}. This
definition admits a generic proof of Lemma \ref{lem:fresh-subst}.

However, $\subst{}{}$ is not the main operation of locally
nameless. That distinction belongs instead to an operation called
\emph{opening}, defined in Figure \ref{fig:ln-open}. This operation is
used to define $\beta$-reduction, with the $\beta$-conversion rule
taking the form $(\lambda t) u =_\beta t^u$. Here, $t^u$ stands for
the opening of $t$ by $u$, defined by replacing all indices in $t$
previously bound to the outermost $\lambda$ with $u$. (Note that the
replaced variables are actually de Bruijn indices rather than free
variables, hence this is not a substitution of atoms.) Unlike with
atoms, the replaced indices do not have to share a common
representation, as the representation of an index bound to the outer
lambda depends on how many other abstractions are in scope at the
occurrence---both $0$ and $1$ in $\lambda(0\lambda 1)$ point to the
outermost $\lambda$, for instance. Therefore $\open$ is defined with
an auxiliary function that maintains a count of how many binders we
have gone under during recursion. In order to define operations that
maintain an ``accumulator'' argument like this, we introduce decorated
monads in Section \ref{subsec:decorations}.

As a final example, some locally nameless terms, e.g. $\lambda(0 1)$,
do not correspond to ordinary lambda terms because they have indices
(in this example, the $1$) that do not ``point'' to any
abstraction. Therefore one restricts attention to terms that are
\emph{locally closed}, defined in Figure \ref{fig:ln-lc}. Like
$\open$, $\LC$ is defined with a helper function that counts the
number of binders gone under during recursion. Unlike $\open$, $\LC$
computes a boolean ($\mathbf{2} = \{\top, \bot\}$) instead of a
term. To define and reason about $\LC$, one must integrate both
concepts above to define decorated-traversable functors and
DTMs. $\SigmaLambda{B}$ is an example of a decorated-traversable
functor, and $\term{B}$ is a DTM. As we have shown with Tealeaves
\cite{2023:tealeaves}, this abstraction suffices to prove a large
suite of infrastructural lemmas about the operations above.

\begin{figure}
	\begin{subfigure}{1\linewidth}
  \begin{subfigure}[b]{0.55\linewidth}
    \centering
    \begin{gather*}
      \open : \term{\mathbf{1}} (\mathbb{A} + \mathbb{N}) \to \term{\mathbf{1}} (\mathbb{A} + \mathbb{N}) \to \term{\mathbf{1}} (\mathbb{A} + \mathbb{N})\\
      \open\ u\ t = \mathrm{open}_0\ u\ t
    \end{gather*}
  \end{subfigure}
  \begin{subfigure}[b]{0.43\linewidth}
    \centering
    \begin{gather*}
    	\LC : \term{\mathbf{1}} (\mathbb{A} + \mathbb{N}) \to \mathbf{2} \\
    	\LC\ t = \textrm{LC}_0\ t
    \end{gather*}
  \end{subfigure}
  \end{subfigure}
  \begin{subfigure}{1\linewidth}
  	\begin{subfigure}[b]{0.5\linewidth}
    \[
    \begin{array}{ll}
      \mathrm{open}_n\, u\, \left(\Var\, v\right) =
      \begin{cases}
	u & \textrm{if $v = \bvar{}\, n$} \\
        \Var\, v & \mathrm{else}
      \end{cases} \\
      \mathrm{open}_n\, u\, \left(\App\, t_1\, t_2\right) = \App\, \left(\mathrm{open}_n\, u\, t_1\right) \left(\mathrm{open}_n\, u\, t_2\right)\\
      \mathrm{open}_n\, u\, \left(\Abs\, \star\, t\right) =  \Abs\, \star\, \left(\mathrm{open}_{n + 1}\, u\, t\right)
    \end{array}
    \]
    \caption{Opening a lambda term by $u$}
    \label{fig:ln-open}
  \end{subfigure}
  \begin{subfigure}[b]{0.48\linewidth}
  \centering
    \[
    \begin{array}{ll}
      \mathrm{LC}_n \left(\Var\ v\right) =
      \begin{cases}
	\bot & \textrm{if $v = \bvar{}\ m$ and $n \leq m$} \\
        \top & \mathrm{else}
      \end{cases} \\
      \mathrm{LC}_n \left(\App\ t_1\ t_2\right) = \mathrm{LC}_n\ t_1\ \land\ \mathrm{LC}_n\  t_2\\
      \mathrm{LC}_n \left(\Abs\ \star\ t\right) =  \mathrm{LC}_{n + 1}\ t
    \end{array}
    \]
    \caption{Testing for local closure}
    \label{fig:ln-lc}
  \end{subfigure}
  \end{subfigure}
  \caption{Operations on locally nameless terms}
  \label{fig:ln-ops}
\end{figure}

\section{Decorated Traversable Functors}
\label{sec:decorated-traversable-functors}

We introduce decorated and traversable monads separately before
incorporating both to form DTMs. We present definitions
type-theoretically alongside a diagrammatic calculus. For ease of
reading, the different sorts of wires in our graphical calculus, which
play different roles, are typeset with high-contrast colors.

\subsection{Decorations}
\label{subsec:decorations}
The category $\Dec{W}$ (Definition \ref{def:decoratedfunctor}) of
decorated functors is parameterized by some monoid $W$, which we take
as given. In Tealeaves, $W$ is typically the free monoid
$\coqlist\ B$, representing the list of the binders in scope at some
occurrence. In brief, decorated functors arise from the elementary
fact that any monoid $W$ in $\Set$ forms a unique \emph{bi}monoid---a
coherent combination of a monoid and a comonoid on the same set. The
``product-with'' embedding,
\[
X \mapsto \left(X \times -\right) : \Set \to \EndSet
\]
is strong monoidal,\footnote{As opposed to merely lax or oplax
monoidal, not to be confused with tensorial strength.} meaning it
preserves monoids, comonoids, and indeed bimonoids, making
$(\writerLong)$ a \emph{bimonad}.  Decorated functors are precisely
the right comodules of this bimonad, which, by
adapting a construction from abstract algebra (see Section 4.1 of
\cite{2007:bakkehopfalgebras}), form a monoidal
category. This means we can consider monoids of decorated functors, or
decorated monads. Now we step throw this slowly.

As a first step, consider any set $E$. It is an exercise in
definitions to verify that \(E\) is the carrier of exactly one
comonoid, the duplication comonoid over \(E\).  This structure
captures aspects of classical information and its fundamental
operations of duplication and deletion.

\begin{definition} The \emph{duplication comonoid over
$E:\textrm{Set}$} is given by the following operations.
\[
\begin{array}{rclrcl}
\del{}{} &:& E \to \mathbf{1} & \quad \del{}{}\ e &=& \star\\
\Delta &:& E \to E \times E       & \quad \Delta\ e &=& (e, e)
\end{array}
\]
\end{definition}

The duplication comonoid induces a comonad on $\left(E \times -
\right)$ known to functional programmers as the environment or reader
comonad. In this paper, these wires, which we think of as carrying
``contextual'' information, are drawn in red.

\begin{definition} The \emph{environment comonad over
  $E:\textrm{Set}$} is given by the product functor $(E \times -)$
  equipped with the following operations of \emph{extraction} and \emph{duplication}.

  \noindent\begin{minipage}[b]{.38\linewidth}
  \centering
  \sdiagram{0.6}{tikz/comonad/ops/counit}
  \begin{gather}%
    \extr{\typE\times}{} :  \forall \left(A : \Set\right),\ E \times A \to A \notag \\
    \extr{\envShort}{A} (e, a) = a \label{eqn:extr-def}
  \end{gather}%
  \end{minipage}
  \noindent\begin{minipage}[b]{.6\linewidth}
  \centering
  \sdiagram{0.6}{tikz/comonad/ops/comultiplication}
  \begin{gather}
    \dup{\typE\times}{}  :  \forall \left(A : \Set\right),\ E \times A \to E \times \left(E \times A\right) \notag \\
    \dup{\envShort}{A} (e, a) = (e, (e, a)) \label{eqn:dup-def}
  \end{gather}
  \end{minipage}
\end{definition}

The co-Kleisli arrows of the environment comonad have the form $E
\times A \to B$. In functional programming, this comonad captures
computations $A \to B$ that additionally can read, but not modify, an
environment of type $E$, such as a user-supplied configuration
file. This is a classic example of the general intuition that while
monads can be used to structure computations with ``effects'',
comonads represent notions of computation that depend on a
``context'' \cite{2008:comonadic-computation}.

Now consider our monoid \(W = \langle W, \cdot, 1_W\rangle\). The
duplication comonoid exists on the underlying set of $W$, so in
particular $(W \times -)$ is an instance of the reader
comonad. Additionally, mirroring the comonoid structure, the monoid on
\(W\) gives rise to a monad structure on \((W \times -)\) known
variously as the writer or logger monad.

\begin{definition} The \emph{writer monad} over
  $W:\textrm{Set}$ is given by the product functor $(W \times -)$
  equipped with the following operations.

  \noindent\begin{minipage}[t]{.45\linewidth}
  \centering
  \sdiagram{0.6}{tikz/writer/ops/unit}
  \begin{gather}%
    \ret{\typW\times}{} :  \forall \left(A : \Set\right),\ A \to W \times A \notag \\
    \ret{\writer}{A} a = \left(1_W, a\right)
  \end{gather}%
  \end{minipage}
  \noindent\begin{minipage}[t]{.54\linewidth}
  \centering
  \sdiagram{0.6}{tikz/writer/ops/multiplication}
  \begin{gather}
    \join{\typW\times}{}  :  \forall \left(A : \Set\right),\ W \times \left(W \times A\right) \to W \times A \notag \\
    \join{\writer}{A} (w_1, (w_2, a)) = (w_1 \cdot w_2, a)
  \end{gather}
  \end{minipage}
\end{definition}

If one thinks about functors as functional data structures, then
``decorated'' funtors are ones whose elements each occur in a context
of type $W$.

\begin{definition}
  \label{def:decoratedfunctor}
  \showsubsfalse
  A \emph{decorated functor} $T : \Set \to \Set$ is a right coalgebra
  of the writer bimonad $(\writerLong)$. Explicitly, it is a functor
  equipped with a natural transformation

  \stringDiagramOperationSide{decfun/components/comultiplication}{\dec{\funT}{} : \forall \left(A : \Set\right), T A \to T \left(W \times A\right)}
  subject to the following two laws:

  \stringdiagramequationParam{0.39}{.59}{0.55}{decfun/axioms/identity_r}{\axiomdecorationextract}{dec-extract}

  \stringdiagramequationParam{0.39}{.59}{0.55}{decfun/axioms/associativity}{\axiomdecorationduplicate}{dec-assoc}
\end{definition}

Intuitively, \eqref{eqn:dec-extract-string} states that computing the
context of every element and immediately deleting it is the same as
doing nothing. \eqref{eqn:dec-assoc-string} states that computing each
context once and making a copy of it is the same as computing each
context twice.

\begin{example}
  \label{example:lambda-dec}
  The functor $\SigmaLambda{B}$ is decorated by $\coqlist\, B$, the
  free monoid over $B$. The operation is defined as follows (where by
  abuse of notation we give constructors of $\SigmaLambda{}$ the same
  name as corresponding constructors of \coq{term}):
  \[
  \dec{}{X} : \SigmaLambda{B} X \to  \SigmaLambda{B} \left(\coqlist\, B \times X\right)
  \]
  \[
  \begin{array}{lcl}
    \dec{}{} \left(\App\ x_1\ x_2\right) &=& \App\ \left([], x_1\right)\ \left([], x_2\right) \\
    \dec{}{} \left(\Abs\ b\ x\right)  &=& \Abs\ b\ ([b],x)
  \end{array}
  \]
  \emph{Notation}: $[]$ is the empty list, while $[b]$ is a singleton.
\end{example}

The decoration in Example \ref{example:lambda-dec} encodes the policy
determining which constructors act as binders in which arguments. The
policy states that an abstraction $\lambda b . x$ adds $b$ to the
binding context of all occurrences in its body, but applications
contribute nothing to the binding context of variables.

Technically, we have not yet used the monoid structure assumed of
$W$. A related fact is that we have only defined decorated
\emph{functors}, but our term functor $T$ is a monad. How should these
structures be related to each other? The answer comes from the
recognition that decorated functors form a monoidal category much like
$\EndSet$.

\begin{lemma}
  \label{lem:dec-fun-cat}
  \showsubsfalse
  The category $\Dec{W}$ of decorated functors is given by the following data:
  \begin{itemize}
  \item Objects are endofuntors $T: \Set \to \Set$ paired with a decoration
  \item Morphisms are natural transformations $T_1 \Rightarrow T_2$ that commute with the decorations of $T_1$ and $T_2$
    \stringdiagramequationParam{0.5}{.48}{0.55}{decfun/category/hom/hom}{\defDecoratedMorphism}{def-dec-morphism}
  \end{itemize}
\end{lemma}

That this constitutes a category is clear. Slightly less obvious is
that $\mathbf{Dec}_{W}$ is a strict monoidal category. Like $\EndSet$,
the tensor operation is composition of functors, but we must explain
how to decorate the composition. Likewise, the tensor unit is the
identity functor, whose decoration must also be defined.
\begin{lemma}
  \label{lem:dec-fun-cat-tensor}
  \showsubsfalse
  $\mathbf{Dec}_{W}$ is a monoidal category by the following data:
  \begin{itemize}
  \item The tensor unit is the identity functor paired with the ``null'' decoration

    \stringdiagramequationBig{decfun/category/tensor/definition/unit/identity}{\defDecoratedId}{def-dec-identity}

  \item Tensor product is given by composition of functors, with decorations added monoidally

    \stringdiagramequationBig{decfun/category/tensor/definition/composition/op}{\defDecoratedComp}{def-dec-comp2}

  \end{itemize}
\end{lemma}

Above, ${\showsubsfalse\strength{\funT_2}{\typW}{}}: \forall \left(A :
\Set\right), W \times T_2\, A \to T_2 \left(W \times A\right)$ is the
tensorial strength operation, depicted as crossing a red wire over a
functor. That \eqref{eqn:def-dec-identity-string} and
\eqref{eqn:def-dec-comp2-string} satisfy axioms
\eqref{eqn:dec-extract-string}---\eqref{eqn:dec-assoc-string} is
easily verified, as are the laws governing the tensor operation.

Since \(\Dec{W}\) is a monoidal category, it makes sense to consider
monoids in this category. Such a structure must be both an ordinary
monad and a decorated functor. The new detail is that the monad
operations must also satisfy \eqref{eqn:def-dec-morphism-string},
given the operations defined in Lemma
\ref{lem:dec-fun-cat-tensor}. This yields two additional equations.

\begin{definition}
  \label{def:decoratedmonad}
  \showsubsfalse
  A \emph{decorated monad} is a monoid in $\Dec{W}$. Explicitly, it is
  equipped with the structures of both a decorated functor and a monad
  such that the following equations are also satisfied.

  \noindent\begin{minipage}{.43\linewidth}
  \centering
  \sdiagram{0.55}{tikz/decmon/axioms/unit}
  \begin{equation}
    \dec{\funT}{\typA} \comp \ret{\funT}{\typA}  = \ret{\funT}{\typW\times\typA} \comp \ret{\writerShort}{\typA}
    \label{eqn:dec-ret}
  \end{equation}
  \end{minipage}%
  \begin{minipage}{.55\linewidth}
  \centering
  \sdiagram{0.53}{tikz/decmon/axioms/butterfly}
  \begin{equation}
    \dec{\funT}{\typA} \comp \join{\funT}{\typA} = \join{\funT \comp \writerShort}{} \comp \dec{\funT}{\funT(\typW\times\typA)} \comp \map{\funT}{} \left(\dec{\funT}{\typA}\right)
    \label{eqn:dec-join}
  \end{equation}
  \end{minipage}
  \vspace{0.3cm}
\end{definition}

In \eqref{eqn:dec-join}, $\join{\funT \comp \writerShort}{}$ is an
abbreviation for
\begin{equation*}
  \showsubsfalse
\join{\funT}{} \comp \map{\funT}{} \left(\map{\funT}{}\, \left(\join{\writerShort}{}\right) \comp \strength{\funT}{\typW}{}\right) : \forall (A : \Set), T (W \times T (W \times A)) \to T (W \times A)
\end{equation*} Indeed, this operation is part of a monad instance on $T \comp (W \times -)$.

In the context of syntax metatheory, \eqref{eqn:dec-ret} states that
an atomic term (some $\Var\ x$) has no binders---the context of $x$ is
the monoid unit, typically the empty list or the natural number
$0$. \eqref{eqn:dec-join} governs how decoration behaves when we
compose constructors to form complex syntax trees. It states that the
context of each variable instance is the concatenation of the context
contributed by each constructor. That is, binders accumulate as one
recurses down a syntax tree, as in the recursive operations from
Figure \ref{fig:ln-ops}.

\begin{example}
  \label{example:lam-decorated-monad}
  The monad $\term{B}$ is decorated by $\coqlist\, B$. The operation
  annotates each variable with the list of $B$ values encountered on
  the unique path from root of the syntax tree to the variable
  occurrence. We show examples using fully named and locally nameless
  variables:

  \begin{minipage}[t]{0.46\linewidth}
    \vspace{-\baselineskip}
    \begin{gather*}
      \dec{}{} : \term{\mathbb{A}}\, \mathbb{A} \to \term{\mathbb{A}}\ \left(\coq{list}\, \mathbb{A} \times \mathbb{A}\right) \\
      \lambda x. \lambda y. y x \mapsto \lambda x . \lambda y. ([x, y], y) ([x, y], x) \\
      \left(\lambda x. y \lambda y. z \right) \mapsto \left(\lambda x. ([x], y) \lambda y. ([x, y], z) \right)
    \end{gather*}
  \end{minipage}%
  \begin{minipage}[t]{0.49\linewidth}
    \vspace{-\baselineskip}
    \begin{gather*}
      \dec{}{} : \term{\mathbf{1}}\, \left(\mathbb{A} + \mathbb{N}\right) \to \term{\mathbf{1}}\, (\mathbb{N} \times \left(\mathbb{A} + \mathbb{N}\right)) \\
      \lambda\lambda 01 \mapsto \lambda\lambda (2, 0) (2, 1) \\
      \left(\lambda 0\right) \left(\lambda \lambda 1 \right) \mapsto      \left(\lambda (1, 0)\right) \left(\lambda \lambda (2, 1) \right) \\
    \end{gather*}
  \end{minipage}

  Note that in the locally nameless example we make the implicit
  identification $\coq{list}\, \mathbf{1} \simeq \mathbb{N}$.
\end{example}

The payoff of this definition will be explained after we consider the
separate issue of traversability.

\subsection{Traversals}
\label{subsec:traversals}

Intuitively, a traversable data structure is a finitary container we
can ``iterate'' \cite{2009:essenseiteratorpattern} over, such as a
\coq{list} or \coq{tree} type. McBride and Paterson
\cite{2008:applicativeprogramming} defined traversable functors as
those equipped with a distributive law over applicative functors
(i.e. lax monoidal endofunctors on $\Set$). Subsequent work
\cite{2009:essenseiteratorpattern, 2012:traversals} refined the notion
by supplying an appropriate set of axioms for this operation.

\begin{definition}
  \label{def:applicative}
  An \emph{applicative functor} is a set-forming operation \mbox{$F : \Set \to \Set$} with operations
  \begin{align*}
    \pure{\funF}{} &: \forall \left(A : \Set\right),~A \to F A \\
    (\circledast)^{\funF} &: \forall \left(A\, B : \Set\right),~F \left(A \to B\right) \to F A \to F B
  \end{align*}
  subject to the following equations (note that $\circledast$ is left-associative).

  \noindent\begin{minipage}[t]{0.48\linewidth}
  \vspace{-\baselineskip}
  \begin{gather}
    \pure{}{}\ \id{} \circledast \fn{a} = \fn{a} \label{eqn:applicative-1} \\
    \pure{}{}\ \fn{f} \circledast \pure{}{}\ \fn{a} = \pure{}{}\ \left(f a\right) \label{eqn:applicative-2}
  \end{gather}
  \end{minipage}%
  \begin{minipage}[t]{0.5\linewidth}
    \vspace{-\baselineskip}
    \begin{gather}
      g \circledast \left(f \circledast a\right) = \pure{}{}\ \left(\cdot\right) \circledast g \circledast \fn{f} \circledast \fn{a} \label{eqn:applicative-3} \\
      \fn{f} \circledast \pure{}{}\ \fn{a} = \pure{}{}\ \left(f \mapsto \fn{f} a\right) \circledast \fn{f} \label{eqn:applicative-4}
    \end{gather}
  \end{minipage}
\end{definition}

This class includes the identity functor $\idfun$ and is closed under
composition. An important special case are constant applicatives:
these must map all sets to some monoid $M$, with the operations and
axioms coinciding with those of monoids.

\begin{definition}
  \label{def:applicative-morphism}
  An \emph{applicative morphism} $\phi : F \Rightarrow G$ is a natural
  transformation between applicative functors that commutes with
  $\pure{}{}$ and $\left(\circledast\right)$ in an obvious way.
\end{definition}

Traversable functors are those that distribute over any choice of
applicative functor in a well-behaved way.

\begin{definition}
  \label{def:traversable-functor-strings}
  \showsubsfalse
  A traversable functor is equipped with an operation
  \begin{figure}[H]
    \captionsetup{labelformat=empty} \centering
    \sdiagram{0.8}{tikz/travfun/components/distribute}
    \caption{$\dist{\funT}{}{} : \forall  \left(F : \mathrm{Applicative}\right) \left(A : \Set\right), T (F A) \to F (T A)$}
  \end{figure}
  subject to the following axioms ($\phi$ ranging over applicative morphisms).

  \stringdiagramequationParam{0.44}{.54}{0.55}{travfun/axioms/identity}{\axiomDistUnit}{dist-id}

  \stringdiagramequationParam{0.47}{.52}{0.54}{travfun/axioms/composition}{\axiomDistComp}{dist-compose}

  \stringdiagramequationParam{0.44}{.54}{0.6}{travfun/axioms/homomorphism}{\axiomDistMorphism}{dist-hom}
\end{definition}

The connection between traversability and container-like properties is
best exemplified by choosing $F$ to be a constant functor over a
monoid $M$. Then, the type of ${\showsubsfalse\dist{}{}{}}$ reduces to $T M \to
M$. Intuitively, $T$ contains a finite number of elements, so that
when all elements have type $M$, we can combine them together using
multiplication in $M$. Gibbons and Oliveira
\cite{2009:essenseiteratorpattern} pointed out that
\eqref{eqn:dist-id-string} forbids this operation from ``skipping''
any elements in $T$, while Jaskelioff and Rypacek
\cite{2012:traversals} pointed out that
\eqref{eqn:dist-compose-string} forbids this operation from ``double
counting'' any elements.

Waern \cite{2019:cofree-traversables} defined the monoidal category of
traversable functors. An arrow in this category is a natural
transformation between traversable functors that commutes with
$\dist{}{}{}$ in an obvious way.

\begin{lemma}[Category $\Trav$]
  \label{lem:trav-fun-cat}
  \showsubsfalse
  The category $\Trav$ of traversable functors is given by the following data:
  \begin{itemize}
  \item Objects are endofunctors $T: \Set \to \Set$ paired with a distributive law over applicative functors
  \item Morphisms are natural transformations $\psi : T_1 \Rightarrow T_2$ that commute with distribution.

    \stringdiagramequationParam{0.49}{.49}{0.55}{travfun/category/hom/hom}{\traversableMorphism}{trav-morph}
  \end{itemize}
\end{lemma}

The identity functor is trivially traversable, and the composition of traversables is traversable just by composing the distributions. Hence, traversable functors $\Trav$ forms a monoidal category. As before, we can consider monoids in this category. These are monads that are also
traversable and whose monad operations satisfy \eqref{eqn:trav-morph-string}.

\begin{definition}
  \label{def:traversable-monad}
  \showsubsfalse A \emph{traversable monad} $T$ is a monoid in
  $\Trav$. Explicitly, $T$ has the structures of both a traversable
  functor and a monad and satisfies the following equations:

  \stringdiagramequation{travmon/axioms/unit}{\axiomDistRet}{dist-ret}

  \stringdiagramequation{travmon/axioms/join}{\axiomDistJoin}{dist-join}
\end{definition}

Though the laws appear opaque, for syntax metatheory,
\eqref{eqn:dist-ret-string} states that a term formed from
$\ret{}{}/\Var$ contains only a single
variable. \eqref{eqn:dist-join-string} implies that substituting a
subterm $u$ for $x$ in $t$ adds the occurrences in $u$ to the set of
occurrences of $t$. This concept is more thoroughly examined in
\cite{2023:tealeaves}.

\subsection{Decorated Traversable Functors}
\label{subsec:decorated-traversable-functors}
For functors that are both traversable and decorated, it is necessary
to impose one more condition relating the decoration and distribution
operations. For the following definition, we note that
$\left(\writerLong\right)$ is uniquely traversable.
\begin{definition}
  \label{def:dtf}
  \showsubsfalse A \emph{decorated-traversable} functor is equipped
  with the structure of both a decorated and traversable functor (Definitions \ref{def:decoratedfunctor} and \ref{def:traversable-functor-strings}),
  subject to the following extra condition:

  \stringdiagramequationBig{dtf/axioms/law}{\axiomDecorationTraverse}{dec-trav}
\end{definition}

\begin{lemma}[Category $\DecTrav{W}$]
  \label{def:dectrav}
  \showsubsfalse
  The strict monoidal category $\DecTrav{W}$ of decorated-traversable functors is given by the following data:
  \begin{itemize}
  \item Objects are decorated traversable functors
  \item Morphisms are natural transformations satisfying both \eqref{eqn:def-dec-morphism-string} and \eqref{eqn:trav-morph-string}.
  \item The tensor product is given by composition of decorated-traversable functors, with the identity functor serving as the tensor unit.
  \end{itemize}
\end{lemma}

\begin{definition}
  \label{def:dtm}
  A \emph{decorated traversable monad} (DTM) is a monoid in $\DecTrav{W}$.
\end{definition}

The force of Definition \ref{def:dtm} is that a DTM is simultaneously
an instance of Definitions \ref{def:decoratedmonad},
\ref{def:traversable-monad}, and \ref{def:dectrav}. A self-contained
summary of the axioms can be found in the appendix.

\section{Kleisli Representation for DTMs}
\label{sec:kleisli}

Definition \ref{def:dtm} is phrased in terms of principled categorical
abstractions, but this is not the most convenient presentation when
working in a theorem prover. Just proving that a syntax forms a DTM is
tedious, requiring five operations and 19 equations. The following
Kleisli-style definition, mirroring Definition \ref{def:kleislimonad},
is more economical and more useful to program with.

\begin{definition}[DTMs, Kleisli-style]
  \label{def:dtmkleisli}
  \showsubsfalse
A \emph{Kleisli-presented DTM} is a set-forming operation $T$ equipped with two operations of the following types
\[
\begin{array}{rclrcl}
\ret{}{} &:&  \forall \left(A : \Set\right),\ A \to T A \\
\binddt{}{}{} &:& \forall \left(F : \mathrm{Applicative}\right) \left(A : \Set\right),\ \left(W \times A \to F (T B)\right) \to T A \to F \left(T B\right)
\end{array}
\]
subject to the following laws (where $\phi$ is quantified over applicative morphisms $\phi : F \Rightarrow G$)
\begin{gather}
  \binddt{}{\idfun}{}\ \left(\ret{\funT}{} \comp \extr{\writerShort}{}\right) = \id{\funT \typA} \label{eqn:binddt-1} \\
  \binddt{}{\funF}{}\ f \comp \ret{\funT}{} = f \comp \ret{\writer}{} \label{eqn:binddt-2} \\
  \map{\funF}{}\ \left(\binddt{}{\funG}{}\ g\right) \comp  \left(\binddt{}{\funF}{}\ f\right) =
  \binddt{}{\funF \comp \funG}{}\ \left( \lambda (w, a) . \map{\funF}{}\ \left(\binddt{}{\funG}{}\ \left(g \odot w\right)\right) f (w, a)\right) \label{eqn:binddt-3} \\
  \phi \comp \binddt{}{\funF}{}\ f = \binddt{}{\funG}\ \left(\phi \comp f\right)  \label{eqn:binddt-4}
\end{gather}
In \eqref{eqn:binddt-3}, $(\odot)$ is defined $(g \odot w_1)\, (w_2, b) \overset{def}{=} g\ (w_1 \cdot w_2, b)$.
\end{definition}

The following theorem speaks to the robustness of Definition \ref{def:dtmkleisli}.

\begin{theorem}
  \label{thm:dtmtokleisli}
  Definitions \ref{def:dtm} and \ref{def:dtmkleisli} are equivalent.
\end{theorem}
\begin{proof}
  \showsubsfalse
  The $\ret{}{}$ operation is the same for both presentations. Given
  $\map{}{}$, $\join{}{}$, $\dec{}{}$, and ${\dist{}{}{}}$, we define
  $\binddt{}{}{}$ as follows:
  \begin{equation}
    \label{eqn:binddt-derived}
    \binddt{}{\funF}{}\, f \overset{def}{=} \map{\funF}{} \left(\join{\funT}{}\right) \comp \dist{\funT}{\funF}{} \comp \map{\funT}{}\ f \comp \dec{}{}.
  \end{equation}
  Given $\ret{}{}$ and $\binddt{}{}{}$, we define the operations of DTMs
thus:
\[
\begin{array}{rclrcl}
\map{}{} f &\overset{def}{=}& \binddt{}{\idfun}{}     \left(\ret{\funT}{} \comp f \comp \extr{\writer}{}\right) &\quad
\dec{}{}   &\overset{def}{=}& \binddt{}{\idfun}{}     \left(\ret{\funT}{}\right)\\
\join{}{}  &\overset{def}{=}& \binddt{}{\idfun}{}     \left(\extr{\writer}{}\right) &\quad
\dist{}{}{} &\overset{def}{=}& \binddt{}{\funF}{}      \left(\ret{\funT}{} \comp \extr{\writer}{}\right)
\end{array}
\]
Besides verifying these definitions satisfy the appropriate equations,
that starting with either representation and completing a roundtrip
returns the original set of operations. A full proof of this fact can
be found in our GitHub repository. The appendix contains a
string-diagrammatic derivation of
\eqref{eqn:binddt-1}---\eqref{eqn:binddt-4} (Lemma
\ref{lem:dtm-to-kleisli}).
\end{proof}

\begin{example}
  \label{example:lam-binddt}
  \showsubsfalse
  The $\binddt{}{}{}$ operation for $\term{B}{}$ is defined as follows (for any $f : W \times A \to F (T B)$):
  \[
  \begin{array}{lcl}
    \binddt{}{\funF}{}\ f\ \left(\Var\ v \right)        &=& \fn{f} (\mathrm{[ ]}, v) \\
    \binddt{}{\funF}{}\ f\ \left(\App\ t_1\ t_2\right)  &=& \pure{\funF}{}\ \App\ \circledast \binddt{}{\funF}{}\ f\ t_1 \circledast \binddt{}{\funF}{}\ f\ t_2 \\
    \binddt{}{\funF}{}\ f\ \left(\Abs\ b\ t\right)      &=& \pure{\funF}{}\ \left(\Abs\ b\right)\ \circledast \binddt{}{\funF}{}\ \left(f \odot [b]\right)\ t
  \end{array}
  \]
\end{example}

Like $\bind{}{}$, ${\showsubsfalse\binddt{}{}{}}$ can be seen as a
template for defining structurally recursive operations on abstract
syntax trees. However, it is appreciably more expressive, introducing
two new features. First, the first argument of $f$ is now a list of
binders in scope at each variable. Second, the output of $f$ is
wrapped in an applicative functor, and all function application is
replaced with ``idiomatic'' application ($\circledast$). Incorporating
these aspects greatly expands the range of operations we can define
generically.

\subsection{Substitution Metatheory}
\label{subsec:reasoning-about-substitution}
Figure \ref{fig:ln-ops-generic} contains generic versions of the
opening operation and local closure, relating to Figure
\ref{fig:ln-ops} as Figure \ref{fig:bind-subst} does to
\ref{fig:subst}.  The definition of $\LC$ in particular requires full
use of the expressiveness of $\binddt{}{}{}$. Here, $\mathbf{2}$
stands for the constant applicative functor over the monoid $\langle
\mathbf{2}, \land, \top \rangle$, which provides a form of universal
quantification over variables. As instances of
${\showsubsfalse\binddt{}{}{}}$, we can reason about these operations
axiomatically.

\begin{figure}[h]
  \showsubsfalse
\noindent\begin{minipage}{0.52\linewidth}
\begin{gather*}
  \openloc : T \left(\mathbb{A} + \mathbb{N}\right) \to \mathbb{N} \times \left(\mathbb{A} + \mathbb{N}\right) \to \mathrm{T}\left(\mathbb{A} + \mathbb{N}\right) \\
  \begin{align*}
    \openloc\, u\, \left(n, \fvar\, a\right) &= \ret{}{} \left(\fvar\, a\right)\\
    \openloc\, u\, \left(n, \bvar\, m\right) &= \begin{cases*} u & \textrm{if $n = m$} \\  \ret{\funT}{} \left(\bvar\, \mathrm{m}\right) & \textrm{else}  \end{cases*}
  \end{align*}\\
  \open\, u = \binddt{\mathrm{T}}{\idfun}{}\, \left(\openloc\, u\right)
\end{gather*}
\end{minipage}%
\begin{minipage}{0.47\linewidth}
  \begin{gather*}
    \LCloc : \mathbb{N} \times \left(\mathbb{A} + \mathbb{N}\right) \to \mathbf{2} \\
    \begin{align*}
      \LCloc \left(n, \fvar\, a\right) &= \top \\
      \LCloc \left(n, \bvar\, m\right) &=
      \begin{cases}
        \bot & \textrm{if $n \leq m$} \\
        \top & \mathrm{else}
      \end{cases}
    \end{align*}\\
    \LC = \binddt{\mathrm{T}}{\mathbf{2}}{}\ \LCloc
  \end{gather*}
\end{minipage}
\caption{Generic locally nameless operations for a DTM $T$}
\label{fig:ln-ops-generic}
\end{figure}

The adequacy of Definition \ref{def:dtmkleisli} for the needs of
working metatheorists is an empirical question demonstrated by
formalizing generic syntax metatheory with it. For comparison, Weirich
and Aydemir previously introduced LNgen \cite{2010:lngen}, a code
generator that accepts a grammar and synthesizes files containing
locally nameless infrastructure for it in Coq. Using Tealeaves, we
were able to formalize all of the infrastructure lemmas defined in
\cite{2010:lngen}, as well as others, statically and generically over
a choice of arbitrary DTM. We have not found any lemmas of the locally
nameless representation that we cannot prove in this fashion. The
advantage of Tealeaves over LNgen is that our lemmas are proven once
and for all, while LNgen generates proofs specific to a given
signature. Because it relies on heuristics and Ltac \cite{2000:ltac}
(Coq's incompletely specified proof automation language), the authors
have reported in private correspondence that LNgen can fail to prove
some lemmas. Additionally they have reported long compile times which
must be re-endured after any changes to the user's syntax. These
downsides do not apply to Tealeaves because it is a static Coq library
rather than a program. The cost of entry is to furnish a proof of
\eqref{eqn:binddt-1}--\eqref{eqn:binddt-4}, which we hope to automate
in future work.

We have also developed a generalization of DTMs for languages with
multiple sorts of variables, and re-derived the same locally nameless
infrastructure, now extended to reason about operations affecting
different sorts of variables.

\section{Related Work}
\label{sec:related}

Bellegarde and Hook \cite{1994:bh-formal-monads} first considered term
monads in the context of formal metatheory. They defined substitution
for a de Bruijn encoding in terms of a combinator \coq{Ewp} (``extend
with policy'') which is similar in spirit to, but strictly less
expressive than, $\binddt{}{}{}$. Lacking axioms comparable to
\eqref{eqn:binddt-1}--\eqref{eqn:binddt-4}, they were unable to reason
about substitution generically.

Subsequent work has generally considered intrinsically well-scoped
\cite{1999:monadic-lambda} and well-typed \cite{2000:deBruijnNested,
  2005:mcbride, 2017:type-scope-safe-programs} representations using
heterogeneous datatypes \cite{1998:bird-nested}. Leveraging the
metatheory's type system to constrain object terms will tend to lead
to a more dependently-typed style of programming where operations and
their correctness properties are woven together. Building on this line
of work, Ahrens et al. \cite{2022:ahrens-category-typed-syntax} have
recently proposed an intrinisically typed language formalization
framework in Coq. The goal of Tealeaves is to support raw syntax,
which involves defining operations first and reasoning about them post
factum.

Fiore and collaborators \cite{1999:abstractsyntax, 2008:fiore-soas}
have developed a presheaf-theoretic account of syntax. Subsequent work
by Power and Tanaka axiomatized and expanded the presheaf-theoretic
approach \cite{2003:power-unified-binding,
  2008:power-tanaka-typed-signatures}. The basic idea is that
intrinsically scoped terms are stratified by a context---the set of all
contexts is then used as the indexing category for the presheaves. In
our development, syntax is parameterized by types $V$ and $B$ for
representations of variables and binder annotations. These are fixed
by a particular representation strategy (e.g. locally nameless) and
one is left with a single set of terms rather than a presheaf. Fiore
and Szamozvancev have proposed a intrinsically well-scoped,
well-typed, syntax formalization framework in Agda
\cite{2022:fiore-metatheory-soas} which takes inspiration from the
presheaf approach.

Approaches that differ more dramatically from ours include strategies
based on nominal sets \cite{2002:new-abstract-syntax} and variations
of higher-order abstract syntax \cite{1988:hoas, 2008:phoas}.

Besides LNgen, utilities similar in spirit to Tealeaves include GMeta
\cite{2011:GMeta} and Autosubst \cite{2015autosubst, 2019:autosubst2}.
GMeta is a Coq framework for generic raw, first-order syntax. Like
Tealeaves, it is parameterized by a variable encoding strategy. GMeta
resorts to proofs by induction on a universe of representable types,
while Tealeaves is based on a principled equational theory. Autosubst
is an equational framework for reasoning about de Bruijn indices in
Coq based on explicit substitution calculi
\cite{1991:explicit-substitutions, 2015:de-bruijn-algebra-coq}. Our
$\binddt{}{}{}$ can express de Bruijn substitution; it may be
enlightening to consider DTMs vis-à-vis these calculi.

\section{Conclusion and Future Work}
\label{sec:conclusion}

We have presented decorated traversable monads, an enrichment of
monads on the category of sets that can be used to reason equationally
about raw, first-order representations of variable binding.

As presented, DTMs are not equipped with a binder-renaming operation
necessary to implement a fully named binding strategy. A first step in
this direction is to recognize that $\term{}$ is also a functor in $B$
besides $V$, yielding an operation
\[
\mathrm{bmap} : \forall \left(\fn{V}\ \typB_1\ \typB_2 : \Set\right), \left(B_1 \to B_2\right) \to \term{B_1} V \to \term{B_2} V
\]
We are investigating an extension of DTMs that incorporates the
functor instance in $B$. One intended application is to provide a
certified generic translation between a named and locally nameless
representation, which could be used as part of a certified compiler,
for example.

Imposing a distributive law over all applicative functors imposes an
order on variable occurrences, which may be unnecessarily strong. Some
process calculi, for example, feature a notion of parallel composition
$|$ such that formulas $p_1 | p_2$ and $p_2 | p_1$ should be taken as
syntactically identical. To support quotiented syntax, one might
require a distributive law only over commutative applicative functors.

\vfill
\pagebreak

\bibliography{bibliography}
\bibliographystyle{eptcs}

\appendix
\raggedbottom
\pagebreak

\section{Appendix}
\begin{figure}[H]
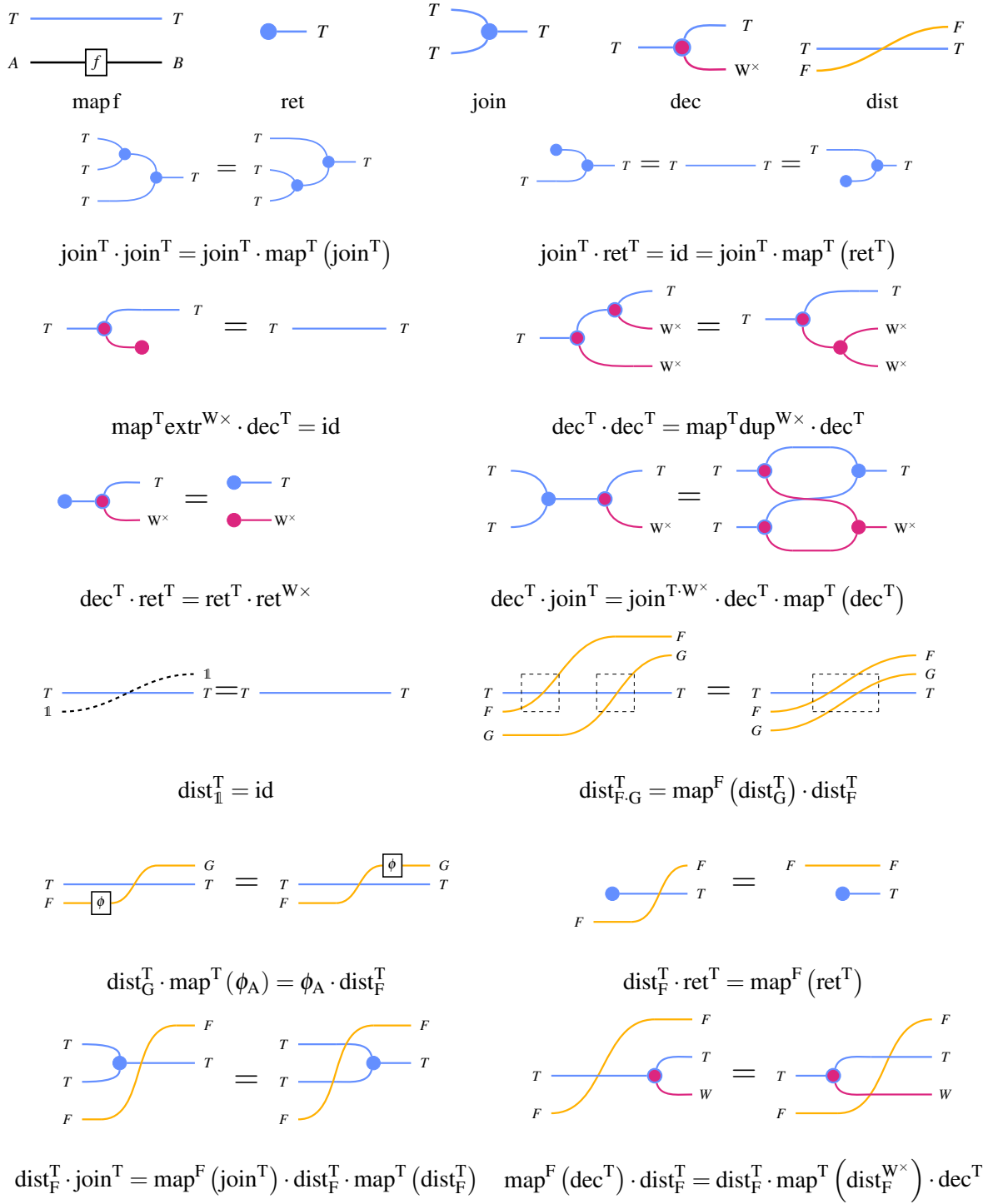

  \showsubsfalse
  \centering
  \begin{subfigure}{1\linewidth}
    \begin{subfigure}{0.19\linewidth}
      \centering
      \captionsetup{labelformat=empty}
      \sdiagram{0.7}{tikz/dtm/misc/ops/map}
      \caption{$\mathrm{map}\, \mathrm{f}$}
    \end{subfigure}
    \begin{subfigure}{0.19\linewidth}
      \centering
      \captionsetup{labelformat=empty}
      \sdiagram{0.7}{tikz/monad/ops/unit}
      \caption{$\mathrm{ret}$}
    \end{subfigure}
    \begin{subfigure}{0.19\linewidth}
      \captionsetup{labelformat=empty}
      \centering
      \sdiagram{0.7}{tikz/monad/ops/multiplication}
      \caption{$\mathrm{join}$}
    \end{subfigure}
    \begin{subfigure}{0.19\linewidth}
      \captionsetup{labelformat=empty}
      \centering
      \sdiagram{0.7}{tikz/dtm/misc/ops/comultiplication}
      \caption{$\mathrm{dec}$}
    \end{subfigure}
    \begin{subfigure}{0.19\linewidth}
      \captionsetup{labelformat=empty}
      \centering
      \sdiagram{0.7}{tikz/dtm/misc/ops/distribute}
      \caption{$\mathrm{dist}$}
    \end{subfigure}
  \end{subfigure}
  \begin{subfigure}{1\linewidth}
    \begin{subfigure}{0.45\linewidth}
      \centering
      \sdiagram{0.5}{tikz/monad/axioms/associativity}
      \[\axiomMonadAssoc\]
    \end{subfigure}
    \begin{subfigure}{0.52\linewidth}
      \centering
      \sdiagram{0.5}{tikz/monad/axioms/identity_all}
      \[\axiomMonadUnitAll\]
    \end{subfigure}
  \end{subfigure}
  \begin{subfigure}{1\linewidth}
    \begin{subfigure}{.45\linewidth}
      \centering
      \sdiagram{0.6}{tikz/dtm/misc/decfun/identity_r}
      \[\axiomdecorationextract\]
    \end{subfigure}
    \begin{subfigure}{.50\linewidth}
      \centering
      \sdiagram{0.6}{tikz/dtm/misc/decfun/associativity}
      \[\axiomdecorationduplicate\]
    \end{subfigure}
  \end{subfigure}
  \begin{subfigure}{1\linewidth}
    \begin{subfigure}{.39\linewidth}
      \centering
      \sdiagram{0.6}{tikz/decmon/axioms/unit}
      \[\axiomDecorationRet\]
    \end{subfigure}
    \begin{subfigure}{.6\linewidth}
      \centering
      \sdiagram{0.6}{tikz/decmon/axioms/butterfly}
      \[\axiomDecorationJoin\]
    \end{subfigure}
  \end{subfigure}
  \begin{subfigure}{1\linewidth}
    \begin{subfigure}{.45\linewidth}
      \centering
      \sdiagram{0.6}{tikz/dtm/misc/trav/identity}
      \[\axiomDistUnit\]
    \end{subfigure}
    \begin{subfigure}{.52\linewidth}
      \centering
      \sdiagram{0.6}{tikz/dtm/misc/trav/composition}
      \[\axiomDistComp\]
    \end{subfigure}
  \end{subfigure}
  \begin{subfigure}{1\linewidth}
    \begin{subfigure}{.49\linewidth}
      \centering
      \sdiagram{0.6}{tikz/dtm/misc/trav/homomorphism}
      \[\axiomDistMorphism\]
    \end{subfigure}
    \begin{subfigure}{.49\linewidth}
      \centering
      \sdiagram{0.6}{tikz/travmon/axioms/unit}
      \[\axiomDistRet\]
    \end{subfigure}
  \end{subfigure}
  \begin{subfigure}{1\linewidth}
    \begin{subfigure}{.49\linewidth}
      \centering
      \sdiagram{0.6}{tikz/travmon/axioms/join}
      \[\axiomDistJoin\]
    \end{subfigure}
    \begin{subfigure}{.49\linewidth}
      \centering
      \sdiagram{0.6}{tikz/dtm/misc/dectrav/law}
      \[\axiomDecorationTraverse\]
    \end{subfigure}
  \end{subfigure}
  \caption{String diagrammatic presentation of DTMs}
\label{fig:dtm-strings}
\end{figure}

\pagebreak
\begin{lemma}
  \label{lem:dtm-to-kleisli}
  \showsubsfalse
  Every DTM gives rise to a Kleisli-presented DTM according to the following definition of $\binddt{}{}{}$.
  \begin{gather*}
    \sdiagram{1}{tikz/dtm/components/binddt} \\
    \binddt{}{\funF}{}\, f = \map{\funF}{} \left(\join{\funT}{}\right) \comp \dist{\funT}{\funF}{} \comp \map{\funT}{}\ f \comp \dec{\funT}{}
  \end{gather*}
\end{lemma}
\begin{proof}
  Proof of Equation \eqref{eqn:binddt-1}:
  \begin{align*}
    & \quad \sdiagram{0.7}{tikz/dtm/kleisli/binddt_ret} \\
    = & \quad \sdiagram{0.7}{tikz/dtm/kleisli/binddt_ret2} && \textrm{\small Apply the decoration cup law \eqref{eqn:dec-ret}.} \\
    = & \quad \sdiagram{0.7}{tikz/dtm/kleisli/binddt_ret3} && \textrm{\small Pull the unit across $F$ \eqref{eqn:dist-ret-string}.} \\
    = & \quad \sdiagram{0.7}{tikz/dtm/kleisli/binddt_ret4} && \textrm{\small Apply the left monad unit law \eqref{eqn:monad-id-l-string}.}
  \end{align*}

  Proof of Equation \eqref{eqn:binddt-2}:
  \begin{align*}
    & \quad \sdiagram{0.7}{tikz/dtm/kleisli/binddt_id1} & \\
    = & \quad \sdiagram{0.7}{tikz/dtm/kleisli/binddt_id2} && \textrm{\small Apply unit and counit laws \eqref{eqn:dec-extract-string} \eqref{eqn:monad-id-r-string}.} \\
    = & \quad \sdiagram{0.7}{tikz/dtm/kleisli/binddt_id3} && \textrm{\small Apply traversal unitary law \eqref{eqn:dist-id-string}.}
  \end{align*}

  Proof of Equation \eqref{eqn:binddt-3}:
  \begin{align*}
    & \quad \sdiagram{0.7}{tikz/dtm/kleisli/binddt_composition_law_proof0}  \\
    = & \quad \sdiagram{0.7}{tikz/dtm/kleisli/binddt_composition_law_proof1} && \textrm{\small Apply the butterfly law \eqref{eqn:dec-join}.} \\
    = & \quad \sdiagram{0.7}{tikz/dtm/kleisli/binddt_composition_law_proof2} && \textrm{\small Drag operations past distributions \eqref{eqn:dist-join-string} \eqref{eqn:dec-trav-string}.} \\
    = & \quad \sdiagram{0.7}{tikz/dtm/kleisli/binddt_composition_law_proof3} && \textrm{\small Apply (co)associativity \eqref{eqn:monad-assoc-string} \eqref{eqn:dec-assoc-string}.} \\
    = & \quad \sdiagram{0.7}{tikz/dtm/kleisli/binddt_composition_law_proof4} && \textrm{\small Apply traversal composition law \eqref{eqn:dist-compose-string}.}
  \end{align*}

  Proof of Equation \eqref{eqn:binddt-4}:
  \begin{align*}
    & \quad \sdiagram{0.7}{tikz/dtm/kleisli/binddt_morph1} \\
    = & \quad \sdiagram{0.7}{tikz/dtm/kleisli/binddt_morph2} && \textrm{\small Slide the applicative morphism \eqref{eqn:dist-hom-string}.}
  \end{align*}
\end{proof}

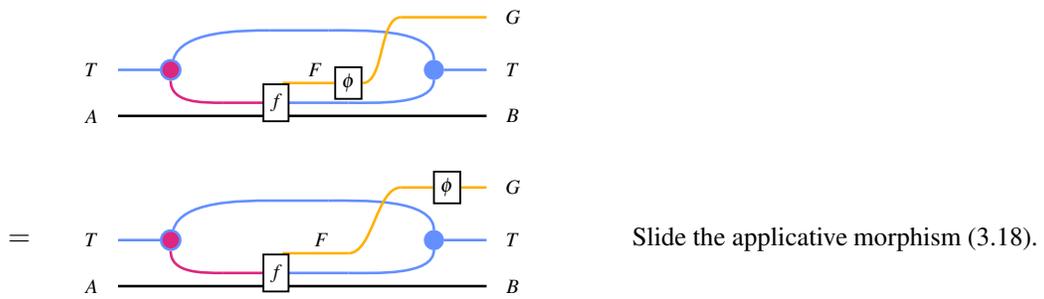
\end{document}